\documentclass[journal,onecolumn]{IEEEtran}
\usepackage{longtable}
\usepackage{multirow}
\usepackage{tabularx}
\usepackage{lipsum}
\usepackage{multicol}
\usepackage{amssymb}
\usepackage{amsmath}
\usepackage{longtable}
\newtheorem{theorem}{Theorem}[section]
\newtheorem{lemma}[theorem]{Lemma}
\newtheorem{proposition}[theorem]{Proposition}

\newtheorem{example}{Example}
\usepackage{xcolor}

\newenvironment{proof}{\begin{IEEEproof}}{\end{IEEEproof}}

\newenvironment{definition}[1][Definition]{\begin{trivlist}
\item[\hskip \labelsep {\bfseries #1}]}{\end{trivlist}}

\newenvironment{remark}[1][Remark]{\begin{trivlist}
\item[\hskip \labelsep {\bfseries #1}]}{\end{trivlist}}

\makeatletter

\newcommand{\Rmnum}[1]{\expandafter\@slowromancap\romannumeral #1@}
\makeatother

\ifCLASSINFOpdf
\else
\fi
\begin{document}
%
\title{Quasi-Perfect Lee Codes from Quadratic Curves over Finite Fields}

\author{ Sihem Mesnager$^1$ \and Chunming Tang$^2$ \and  Yanfeng Qi$^3$
\thanks{This work was supported by
the National Natural Science Foundation of China
(Grant No. 11401480, 11531002). C. Tang
also acknowledges support from 14E013 and
CXTD2014-4 of China West Normal University.
Y. Qi also acknowledges support from
KSY075614050 of Hangzhou Dianzi University.
}
\thanks{S. Mesnager is with Department of Mathematics, Universities of Paris VIII and XIII and Telecom ParisTech, LAGA, UMR 7539, CNRS, Sorbonne Paris Cit\'{e}. e-mail: smesnager@univ-paris8.fr}
\thanks{C. Tang is with School of Mathematics and Information, China West Normal University, Nanchong, Sichuan,  637002, China. e-mail: tangchunmingmath@163.com
}

\thanks{Y. Qi is with School of Science, Hangzhou Dianzi University, Hangzhou, Zhejiang, 310018, China.
e-mail: qiyanfeng07@163.com
}

}

%


\maketitle

\begin{abstract}
Golomb and Welch conjectured in 1970 that there only exist perfect Lee codes for radius $t=1$ or dimension $n=1, 2$. It is admitted that
the existence and the construction of quasi-perfect Lee codes have to be studied since they are the best alternative to the perfect codes.
In this paper we firstly highlight the relationships between subset sums, Cayley graphs, and Lee linear codes and present some results.  Next, we present a new constructive method for constructing quasi-perfect Lee codes. Our approach uses subsets derived from some quadratic curves over finite fields (in odd characteristic) to derive  two classes of $2$-quasi-perfect Lee codes are given over the space $\mathbb{Z}_p^n$ for $n=\frac{p^k+1}{2}$ $(\text{with} ~p\equiv 1, -5 \mod 12 \text{ and } k \text{ is any integer}, \text{ or }  p\equiv -1, 5 \mod 12 \text{ and } k \text{ is an even integer})$ and $n=\frac{p^k-1}{2}$ $(\text{with }p\equiv -1, 5 \mod 12,  k \text{ is an odd integer} \text{ and } p^k>12)$, where $p$ is an  odd prime. Our codes encompass the quasi-perfect Lee codes constructed recently by Camarero and Mart\'{\i}nez. Furthermore, we  solve a conjecture proposed by Camarero and Mart\'{\i}nez
(in "quasi-perfect Lee codes of radius $2$ and arbitrarily large dimension", IEEE Trans. Inf. Theory, vol. 62, no. 3, 2016) by proving that the related Cayley graphs are Ramanujan or almost Ramanujan.  The Lee codes presented in this paper have applications to constrained and partial-response channels, in flash memories and decision diagrams.

\end{abstract}

\begin{IEEEkeywords}
Lee distance, quasi-perfect codes, subset sums, Cayley graphs, Ramanujan graphs, quadratic curves.
\end{IEEEkeywords}

%
\IEEEpeerreviewmaketitle

\section{Introduction}
Throughout this paper, $\mathbb{Z}$ and $\mathbb{Z}_{M}$ denote the ring of integers  and the ring of integer modulo $M$, respectively. By a \emph{Lee code} $\mathcal{C}$ of block size $n$ over $\mathbb{Z}$ (or $\mathbb{Z}_M$) we will understand a subset $\mathcal{C}$ of the infinite lattice $\mathbb{Z}^n$(or the finite integer lattice $\mathbb{Z}_M^n$) endowed by the Lee distance. If $\mathcal{C}$ has furthermore the structure of an additive group, then $\mathcal{C}$ is called \emph{Lee linear code}. Lee codes have many practical applications, in particular for the toroidal interconnection networks. They are also used for phase modulated and multilevel quantized-pulse-modulated channels (see e.g. \cite{Ast82,Ast85,Ber68,CW71} and \cite{RS94}). Moreover, it has been shown in \cite{AlB97,AB98,AKB97,BB97} and \cite{BBKA95} that these codes are the foundation of designing placement strategies to distribute commonly shared resources like input/output (I/O) devices over a toroidal networks. Similar concepts can be used to design fault-tolerant techniques for this kind of networks \cite{BB97}. Using space embeddings, Jiang et al. gave a method  in \cite{JSB10} to construct Charge-Constrained Rank-Modulation codes (CCRM codes) from Lee error-correcting codes, which could be employed for flash memories. Astola and Stankovic  considered in \cite{AS12} Lee codes to build decision diagrams.

Perfect Lee codes are the most interesting and important subclass of Lee codes. Unfortunately, the perfect $t$-error correcting Lee codes of block
length $n$ over $\mathbb{Z}$, and over $\mathbb{Z}_M$, $M\ge 2n+1$, shortly $PL(n,t)$ codes and $PL(n,t,M)$ codes, respectively, have been constructed only for $n=1,2$, and any $t$, and for $n\ge 3$ and $t=1$. Moreover, as suggested by the well-known and long-standing conjecture of Golomb and Welch \cite{GW70}, $PL(n,t)$ codes and $PL(n,t,M)$, $M\ge 2n+1$, codes do not exist in other cases.
Despite the considerable amount of attempts in this topic,  the conjecture is far being solved.

Although the Golomb-Welch conjecture has not been solved yet, its validity has been widely believed by the community. Therefore, failing finding perfect Lee codes, some codes which are "closed" to being perfect have been considered in the literature (see e.g. \cite{AB03}, where quasi-perfect codes over $\mathbb{Z}$ and $QPL(n,t)$ codes have been introduced). Also, in \cite{HG14} the authors presented some quasi-perfect codes for $n=3$ and few radii. Later, Queiroz et al. characterized  in \cite{QCMP13} quasi-perfect codes over Gaussian and Eisenstein-Jacobi integers. As a consequence, linear quasi-perfect Lee codes were obtained for $n=2$. In \cite{CM16}, from the quotient additive group of Gaussian integers, Camarero and Mart\'{i}nez constructed two-quasi-perfect Lee codes over the space $\mathbb{Z}_p^n$ for $p$ prime with $p\equiv \pm 5 \mod 12$, and $n=2 [\frac{p}{4}]$, where the notation $[\frac{p}{4}]$ stands for the closed integer to the rational number $\frac{p}{4}$.

In this manuscript, we firstly survey the connections between subset sums, Cayley graphs, and Lee linear codes. Next, we provide two classes of $2$-quasi-perfect Lee codes over the space $\mathbb{Z}_p^n$ for $n=\frac{p^k+1}{2}$ $(\text{with} ~p\equiv 1, -5 \mod 12 \text{ and } k \text{ is any integer}, \text{ or }\\ p\equiv -1, 5 \mod 12 \text{ and } k \text{ is an even integer})$ and $n=\frac{p^k-1}{2}$ $(\text{if }p\equiv -1, 5 \mod 12,  k \text{ is an odd integer} \text{ and } p^k>12)$, where $p$ is an  odd prime. Our classes are obtained by considering subsets derived from some quadratic curves over finite fields of odd characteristic. This paper generalizes the results of \cite{CM16} since the $2$-quasi-perfect Lee codes constructed by Camarero and Mart\'{\i}nez  in \cite{CM16} correspond in fact to our codes when $p\equiv \pm 5 \mod 12$ and $k=1$. Furthermore, This paper also proves that the related Cayley graphs are Ramanujan or almost Ramanujan. It solves a conjecture proposed by Camarero and Mart\'{\i}nez \cite{CM16}.

This paper is organized as follows. In Section
\ref{sec:preliminaries}, we recall and present some results (useful in the next sections) devoted to
exponential sums, Lee linear codes and Cayley graphs. Next, in
Section \ref{sec:subset sums}, we survey the relationships between subset sums, Cayley graphs, and Lee linear codes. In Section
\ref{sec:subsets}, we introduce two classes of subsets of  $\mathbb{F}_{q^2}$ and $\mathbb{F}_{q}^2$ derived from quadratic curves defined over finite fields $\mathbb{F}_q$ and provide several results related to  their  associate subset sums.
In Section \ref{sec:Cayley}, we show that the related graphs are Ramanujan or almost Ramanujan.
which solves a conjecture proposed by Camarero and Mart\'{\i}nez \cite{CM16}. In Section \ref{sec:quasi-perfect}, we present two infinite classes of
$2$-quasi-perfect Lee codes over finite alphabet.
\section{Preliminaries}\label{sec:preliminaries}

 We begin this section by fixing some notation which will be used throughout this paper.
\begin{itemize}
\item Given a set $S$, $\#	S$ denotes its cardinality;
\item $\mathbb{Z}$ denotes the ring of integers;
\item  $\mathbb{Z}_{M}$ denotes the ring of integers modulo $M$;
\item $\mathbb{F}_{q}$ denotes the finite field
with $q$ elements (where $q$ is a prime power $p^k$);
\item  $\mathbb{F}_q^{\times}$
denotes the multiplicative group of
$\mathbb{F}_q$;
\item Given a prime $p$, $p^*:=(-1)^{\frac{p-1}{2}} p$;
\item $\eta$ and $\eta_0$ denote the quadratic characters of $\mathbb{F}_q^{\times}$ and $\mathbb{F}_p^{\times}$, respectively;
\item $\textrm{SQ}_q$ and $\textrm{N\textrm{SQ}}_q$ denote the set of all  squares and nonsquares in $\mathbb{F}_q^{\times}$, respectively.
\item $\zeta_p=e^{\frac{2\pi\sqrt{-1}}{p}}$ is the primitive $p$-th root of unity in the complex field $\mathbb{C}$.
\end{itemize}
\subsection{Quadratic Reciprocity Law and exponential sums}
Some known results related to particular cases of the well-known \emph{Quadratic Reciprocity Law } (\cite{LN83}) will be necessary for our results. We  summarize those results the following proposition.
\begin{proposition}
If $p$ is an odd prime greater than 3,
then

{\rm (1)} $-1\in \textrm{SQ}_p$ if and only if $p\equiv 1\mod4$.

{\rm (2)} $3\in \textrm{SQ}_p$ if and only if
$p\equiv \pm 1 \mod 12$.

{\rm (3)} $-3\in \textrm{SQ}_p$ if and only if
$p\equiv 1$ or $-5 \mod 12$.
\end{proposition}

We shall also use the following  result related to  quadratic exponential sums \cite{LN83}.
\begin{lemma}\label{cx2a}
Let $q=p^k$, with $p$ an odd prime. Then, for any $a\in \mathbb{F}_q,c\in \mathbb{F}_q^{\times}$, we have
$$
\sum_{x\in \mathbb{F}_q}
\zeta_p^{Tr_1^{k}(cx^2+ax)}
=(-1)^{k-1}\eta(c)
\sqrt{p^*}^k\zeta_p^{Tr_1^k(-\frac{a^2}{
4c})}.
$$
\end{lemma}
 \emph{Kloosterman sums} form a special class of
exponential sums. For $(a,b)\in\mathbb{F}_q\times\mathbb{F}_{q}^\times$, the Kloosterman sum $K_q(a,b)$ is  defined by
$$
K_q(a,b)=\sum_{x\in \mathbb{F}_{q}^\times}\zeta_p^{Tr_1^{k}(
ax+\frac{b}{x})}.
$$
Kloosterman sums are related to several  mathematical and engineering problems and have been extensively studied in the literature. Unfortunately, it is
very hard to evaluate Kloosterman sums. The \emph{Hasse-Weil bound} on Kloosterman sums (\cite{LN83}) is an estimation given by the following statement.
\begin{theorem}\label{Kloost}
Let $(a,b)\in \mathbb{F}^\times_q\times \mathbb{F}^\times_q$. Then
$\vert K_q(a,b)\vert\leq 2\sqrt{q}$.
\end{theorem}

\subsection{Lee linear codes}
Since Lee codes over finite alphabet are the target of our study, the natural space to be considered is the vector space $\mathbb{Z}_p^n$ over primitive finite fields.
Therefore, a code $\mathcal{C}$ will be a subset of $\mathbb{Z}_p^n$. This code is said to be linear if it is a subgroup of $\mathbb{Z}_p^n$. For
$\mathbf{a},\mathbf{b} \in \mathbb{Z}_p^n$ their Lee distance is defined as
$$
d_{L}(\mathbf{a},\mathbf{b})=\sum_{i=1}^n  min \{|s|: s\equiv \mathbf{a}_i-\mathbf{b}_i \mod p, s \in \mathbb{Z}\}.
$$
The weight of a word $\mathbf{c}$ is defined as its distance to the origin $O$, which will be denoted as $wt_{L}(\mathbf{c})=d_{L}(\mathbf{c},O)$.
For any positive integer $r$, the Lee sphere of radius $r$ is defined as all the points whose weight is less or equal to $r$, that is:
$$
B^n_r=\{\mathbf{c}: wt_{L}(\mathbf{c})\le r\}.
$$
Note that, when $p\ge 5$, then, for any dimension $n\ge 1$, we have $\# B^n_2=2n^2+2n+1$ (\cite{GW70}).

A code $\mathcal{C}$ is said to be $t\text{-}error~correcting$ if $t$ is the greatest integer such that for any word $\mathbf{w}$ there is at the most one codeword $\mathbf{c}\in \mathcal{C}$ with $d_{L}(\mathbf{w},\mathbf{c})\le t$. Thus, $t$ is called the $error ~correction $ of $\mathcal{C}$. A code $\mathcal{C}$ is said to be $r\text{-}covering$ if $r$ is the smallest integer such that for any word $\mathbf{w}$ there is at the least  one codeword $\mathbf{c}\in \mathcal{C}$ with $d_{L}(\mathbf{w},\mathbf{c})\le r$. Thus, $r$ denotes the covering radius of $\mathcal{C}$. A code that is both $t\text{-}$error correcting and $t\text{-}$covering is said to be \emph{perfect}. A code that is $t\text{-}$error correcting and $(t+1)\text{-}$covering is said to be $t$-\emph{quasi-perfect}.
Golomb and Welch  conjectured in \cite{GW70} that there only exist perfect Lee codes for $t=1$ or $n=1,2$. Therefore, the existence of quasi-perfect codes must be studied  since they are the best alternative to the prefect codes.

 Linear codes can be constructed from Abelian groups and their special subsets. In the following,   $\Gamma$ denotes an Abelian group with exponential $p$, that is, for any $\beta \in \Gamma $, $p \cdot \beta =0$. Thus, $\Gamma$ can be considered as a vector space over $\mathbb{F}_p$. Set $m:=\mbox{dim}_{\mathbb{F}_p} \Gamma$ and let $H$ be a set of generators of $\Gamma$, with $H=-H$ and $0\not \in H$. Thus, the cardinality $\#H$ of $H$  must be even. By setting, $\# H=2n$ one has $H=\{\pm \beta_1, \pm \beta_2, \cdots, \pm \beta_n \}$ and the associated linear code is
$$
\mathcal{C}(\Gamma; H)=\{(c_1,\cdots,c_n)\in \mathbb{F}_p^n: c_1\beta_1+\cdots +c_n\beta_n=0\}.
$$
The dimension of $\mathcal{C}(\Gamma; H)$ over $\mathbb{F}_p$ equals $n-m$.

\subsection{Cayley graph}
A graph $X$ is a triple consisting of a vertex set $V=V(X)$, an edge set $E=E(X)$ and a map that associates with each edge two vertices (not necessarily distinct) called its \emph{endpoints}. The degree of a vertex $v$ is the number of edges incident with $v$. A graph is called $k\text{-}regular$ if every vertex has degree $k$. Now, let us recall some basic definitions. The distance $d_{X}(x,y)$ between two vertices $x,y$ in a graph $X$ is defined as the number of edges in the shortest path from $x$ to $y$. The diameter of a graph $X$ is the maximum among distances between every pair of vertices.

To any graph, one can associate the $adjacency~ matrix$ $A$ which is an $N\times N$ matrix (where $N=\# V$) with rows and columns indexed by the elements of the vertex set and the $(x,y)$-th entry is the number of edges connecting $x$ and $y$. Since our graphs are undirected, the matrix $A$ is symmetric. Consequently, all of its eigenvalues are real. \emph{Ramanujan graphs} are good expander graphs that attain the spectral bound (\cite{Mur03}). More specifically, a $k$-regular graph $X$ is a Ramanujan graph if and only if, for every eigenvalue $\lambda$ of its adjacency matrix it holds either $|\lambda|=k$ or $|\lambda|\le 2\sqrt{k-1}$.

There is a simple procedure for constructing regular graphs using group theory. This can be described  as follows. Let $\Gamma$ and $H$ be defined as the previous subsection. Now construct the Cayley graph $Cay(\Gamma; H)$ by having the vertex set to be the elements of $\Gamma$ with $(\alpha,\beta)$ an edge if and only if $\beta-\alpha \in H$. Then, $Cay(\Gamma; H)$ is a $\#H$-regular graph.  Then the \emph{error correction capacity} of the Cayley graph $Cay(\Gamma; H)$ is defined as the greatest integer $t$ such that for every vertex $x$ there are $\# B^n_t$ vertices at distance $t$ or less from $x$, where $\# H =2n$.
 Note that Cayley graph is vertex-transitive and therefore it is enough to count the number of vertices around one vertex to determine its error correction capacity and the diameter can also be calculated as the maximum distance to one particular vertex, usually $0\in \Gamma$.

The eigenvalues of the Cayley graph
can be determined by the following theorem given in \cite{Mur03}.
\begin{theorem}\label{eigenvelue}
Let $\Gamma$ be a finite Abelian group and
$H$ a subset of $\Gamma$ with $H=-H$ and $0\not \in H$. The  eigenvalues of
the adjacency matrix of $Cay(\Gamma;H)$
are given by
$$
\lambda_{\chi}=\sum_{\beta \in H}\chi(\beta),
$$
where $\chi$ ranges over all characters of $\Gamma$.
\end{theorem}
\begin{remark}
Let $\# H=2n$. Notice that for the trivial character $\chi_0$, we have $\lambda_{\chi_0}=2n$.  For all
$\chi \neq 1$, one has
$$
|\sum_{\beta\in H}\chi(\beta)|<k.
$$
Then, the graph is connected. Thus, to construct Ranmanujan graphs, we require
$$
|\sum_{\beta\in H}\chi(\beta)|\le 2\sqrt{\#H -1}
$$
for every non-trivial character $\chi$ of $\Gamma$.
\end{remark}
\section{subset sums, Cayley graph and Lee code}\label{sec:subset sums}
The correspondence between subset sums, linear codes and Cayley graph is explained in this section. Firstly, we start by stating some fundamental definitions on subset sums. For any subsets
$A$ and $B$ of an Abelien group,
define $-A=\{-x:x\in A\}$,
$A+B=\{x+y: x\in A,y\in B\}$ and
 $A^{(i+1)}
=A^{(i)}+A$ with $A^{(1)}=A$ and $A^{(0)}=\{0\}$.

We introduce the following notions.
\begin{definition}
Given a finite Abelian group $\Gamma$ and a set of generators $H=\{\pm \beta_1, \cdots, \pm \beta_n\}$ with cardinality $\#H =2n$. The \emph{expansion critical index} of $H$ is defined as the greatest integer $t$ such that $ \# (\cup_{i=0}^{t} H^{(i)}) = \# B^n_t$. The \emph{expansion limit index} of $H$ is defined as the smallest  integer $r$ such that $ \cup_{i=0}^{r} H^{(i)} = \Gamma$.
\end{definition}

\begin{theorem}\label{three quantities 0}
Let $p\ge 5$ be an odd prime, $\Gamma$ be a finite Abelian group with exponential $p$ and $H=\{\pm \beta_1, \cdots, \pm \beta_n \}$ be a set of generators with $\# H=2n$. Then, the following three quantities  coincide,\\
\rm (1) the expansion critical index of subset $H$;\\
\rm (2) the error correction capacity of Cayley graph $Cay(\Gamma;H)$;\\
\rm (3) the error correction of Lee linear code $\mathcal{C}(\Gamma;H)$.
\end{theorem}
\begin{proof}
Let $D_i$ be the set of  vertices  at distance $i$ from $0$. Then, from the definition of Cayley graph $Cay(\Gamma;H)$, the vertices at distance $t$ or less from $0$ is the set $\cup_{i=0}^{t} D_i=\cup _{i=0}^{t} H^{(i)}$. Thus, the expansion critical index of subset $H$ equals the error correction capacity of Cayley graph $Cay(\Gamma;H)$. By Theorem 4 in \cite{CM16}, the error correction capacity of Cayley graph $Cay(\Gamma;H)$ coincides with  the error correction of Lee linear code $\mathcal{C}(\Gamma;H)$, which completes the proof.
\end{proof}

\begin{theorem}\label{three quantities 1}
Let $p\ge 5$ be an odd prime, $\Gamma$ be a finite Abelian group with exponential $p$ and $H=\{\pm \beta_1, \cdots, \pm \beta_n \}$ be a set of generators with $\# H=2n$. Then, the following three quantities coincide,\\
\rm (1) the expansion limit index of subset $H$;\\
\rm (2) the diameter of Cayley graph $Cay(\Gamma;H)$;\\
\rm (3) the covering radius of Lee linear code $\mathcal{C}(\Gamma;H)$.
\end{theorem}
\begin{proof}
Let $D_i$ be defined as in the proof of Theorem \ref{three quantities 0} and the expansion limit index of subset $H$ is $r$. Then, the diameter of Cayley graph $Cay(\Gamma;H)$ is less than or equal to $r$. From the definition of expansion limit index, $(\cup_{i=0}^{r} H^{(i)})\setminus(\cup_{i=0}^{r-1} H^{(i)})\neq \emptyset$. The vertices in $(\cup_{i=0}^{r} H^{(i)})\setminus (\cup_{i=0}^{r-1} H^{(i)})$ has distance $r$ from $0$. Hence, the diameter of Cayley graph $Cay(\Gamma;H)$ is  equal to $r$. By Theorem 4 in \cite{CM16}, the diameter of Cayley graph $Cay(\Gamma;H)$ coincides with the covering radius of Lee linear code $\mathcal{C}(\Gamma;H)$, which completes the proof.
\end{proof}

\begin{theorem}\label{2 and 3}
Let $p\ge 5$ be an odd prime, $\Gamma$ be a finite Abelian group with exponential $p$ and $H=\{\pm \beta_1, \cdots, \pm \beta_n \}$ be a set of generators with $\# H=2n$. If $2n^2+2n+1<\# \Gamma< \frac{1}{3}(1+2n)(3+2n+2n^2)$, $\# (\cup_{i=0}^{2} H^{(i)})=2n^2+2n+1$ and $\cup_{i=0}^{3} H^{(i)}=\Gamma$. Then, the expansion critical index and the expansion limit index of subset $H$ are 2 and 3, respectively.
\end{theorem}
\begin{proof}
Note that the $n$-dimensional sphere of radius $2$ has cardinality $\# B^n_2 =2n^2+2n+1$ and the $n$-dimensional sphere of radius $3$ has cardinality $\# B^n_3 =\frac{1}{3}(1+2n)(3+2n+2n^2)$. From the definitions of expansion critical index and expansion limit index,  the expansion critical index and the expansion limit index of subset $H$ are 2 and 3, respectively. This completes the proof.
\end{proof}

From Theorems \ref{2 and 3}, \ref{three quantities 0} and \ref{three quantities 1}, we derive the following statement.

\begin{theorem}
Let $p\ge 5$ be an odd prime, $\Gamma$ be a finite Abelian group with exponential $p$ and $H=\{\pm \beta_1, \cdots, \pm \beta_n \}$ be a set of generators with $\# H=2n$. If $2n^2+2n+1<\# \Gamma< \frac{1}{3}(1+2n)(3+2n+2n^2)$, $\# (\cup_{i=0}^{2} H^{(i)})=2n^2+2n+1$ and $\cup_{i=0}^{3} H^{(i)}=\Gamma$. Then, $\mathcal{C}(\Gamma;H)$ is a linear $2$-quasi-perfect $p$-ary Lee code over $\mathbb{F}_p^n$ with dimension $n-dim_{\mathbb{F}_p} \Gamma$.

\end{theorem}

The next theorem follows  from Theorems \ref{2 and 3}, \ref{three quantities 0} and  \ref{three quantities 1}.

\begin{theorem}
Let $p\ge 5$ be an odd prime, $\Gamma$ be a finite Abelian group with exponential $p$ and $H=\{\pm \beta_1, \cdots, \pm \beta_n \}$ be a set of generators with $\# H=2n$. If $\# \Gamma=2n^2+2n+1$ and $\cup_{i=0}^{2} H^{(i)}=\Gamma$. Then, $\mathcal{C}(\Gamma;H)$ is a linear $2$-perfect $p$-ary Lee code over $\mathbb{F}_p^n$ with dimension $n-dim_{\mathbb{F}_p} \Gamma$.
\end{theorem}

\section{subsets and subset sums from quadratic curves}\label{sec:subsets}
Let $q=p^k$ with $p$ an odd prime, and
 $\delta$ be a
quadratic nonresidue in $\mathbb{F}_q$.
Then $\mathbb{F}_q[\sqrt{\delta}]$
is an extension of $\mathbb{F}_q$ with degree 2, which is
denoted by $\mathbb{F}_{q^2}$.
For any $z
=x+\sqrt{\delta}y (x,y\in \mathbb{F}_q)$,  define $\overline{z}
=x-\sqrt{\delta}y$ and
$\mathcal{N}(z)=z\cdot \overline{z}
=x^2-\delta y^2$. Then
$\mathcal{N}(\cdot)$ is a surjective morphism from $\mathbb{F}_{q^2}^{\times}
$ to $\mathbb{F}_{q}^{\times}$.
Define a subset of $\mathbb{F}_{q^2}$
\begin{equation}\label{H+}
H_{+}=\{z\in \mathbb{F}_{q^2}:
\mathcal{N}(z)=1
\}.
\end{equation}
Then $\#H_{+}=q+1$ and
$H_{+}$ can also be view as the set of points on the following quadratic curve over
$\mathbb{F}_q$:
$$
x^2-\delta y^2=1.
$$
We  will also considering the following subset of
$\mathbb{F}_q^2$:
\begin{equation}\label{H-}
H_{-}=\{(x,\frac{1}{x}): x\in \mathbb{F}_q^{
\times}\}.
\end{equation}
Then $H_{-}$ is the set of points on the quadratic curve
$$
xy=1.
$$

In the following subsections, we will give some properties of two subsects $H_{+}$ and
$H_{-}$.
\subsection{Some results on subset sums for $H_{+}$}

\begin{lemma}\label{c=q+1}
Let $\delta\in \mathrm{NSQ}_q$ and $c\in \mathbb{F}_q^{\times}$.
Then $\#\{(x,y): x^2-\delta y^2=c
\}=q+1$.
\end{lemma}
\begin{proof}
Note that $x^2-\delta y^2
=c$ if and only if $\mathcal{N}(
x+\sqrt{\delta}y)=c$. Hence, this lemma follows.
\end{proof}

\begin{lemma}\label{D_c}
Let $c\in \mathbb{F}_q^{\times}$.
Define the set $D_c=\{x\in \mathbb{F}_q:
x^2=c ~\text{or}~ x^2-c\in \textrm{NSQ}_q \}$.
Then $\#D_c=
\left\{
  \begin{array}{ll}
    \frac{q+3}{2}, &  c\in \textrm{SQ}_q \\
    \frac{q+1}{2}, &  c\in \textrm{NSQ}_q.
  \end{array}
\right.
$
\end{lemma}
\begin{proof}
Suppose $c\in \textrm{SQ}_q$. From the definition
$D_c$, we have
$$
D_c=\{x\in \mathbb{F}_q: x^2-\delta y^2
=c\}.
$$
Let $\#D_c=s$ and
$D_c=\{\sqrt{c}, -\sqrt{c}, x_1,\cdots, x_{s-2}
\}$. Then all the points  on
$x^2-\delta y^2=c$ are
$$
(\sqrt{c},0),
(-\sqrt{c},0), (x_1,\pm y_1),
\cdots, (x_{s-2},\pm y_{s-2}),
$$
where $y_i\in \mathbb{F}_q^{\times}$ and $x_i^2-\delta y_i^2=c$. Hence, we have
$$
\#\{(x,y): x^2-\delta y^2=c
\}=2(s-2)+2=2s-2.
$$
From Lemma \ref{c=q+1},
$2s-2=q+1$ and $s=\frac{q+3}{2}$.

Suppose $c\in \textrm{NSQ}_q$.
Let $D_c=s$ and $D_c=
\{x_1,\cdots, x_s\}$.
Then all the points on
$x^2-\delta y^2=c$ are
$$
(x_1,\pm y_1), \cdots, (x_s,\pm y_s),
$$
where $y_i\in \mathbb{F}_q^{\times}$ and $x_i^2-\delta y_i^2=c$.
Hence,
$\#\{(x,y): x^2-\delta y^2=c\}=2s$.
From Lemma \ref{c=q+1},
we have $2s=q+1$ and $s=\frac{q+1}{2}$.
\end{proof}

\begin{lemma}\label{Iw}
Let $w\in \mathbb{F}_{q^2}^{\times}$, $c=\mathcal{N}(w)$, and
$I_w=\{\mathcal{N}(z): z\in H_{+}
+H_{+}w
\}$. Then\\
{\rm (1)}
$\#I_w=
\left\{
  \begin{array}{ll}
    \frac{q+3}{2}, & c\in \textrm{SQ}_q \\
    \frac{q+1}{2}, & c\in \textrm{NSQ}_q
  \end{array}
\right.
$.\\
{\rm (2)} $1\in I_1$ if and only if
$-3\in \textrm{NSQ}_q$ or $p=3$.
\end{lemma}
\begin{proof}
\rm(1) Let $z_1,z_2\in H_{+}$. Then
$\mathcal{N}(z_1+z_2 w)=\mathcal{N}(
z_1(1+\frac{z_2}{z_1} w))
=\mathcal{N}(z_1)\mathcal{N}(1+
\frac{z_2}{z_1}w)$.
From the definition of $H_+$, we have
$$
I_w=\{\mathcal{N}(1+z): z\in H_{+}w\}.
$$
From the definition of $\mathcal{N}(\cdot)$,
$t\in I_w$ if and only if the following
system of equations has solutions:
$$
\left\{
  \begin{array}{l}
    x^2-\delta y^2=c  \\
     (1+x)^2-\delta y^2=t.
  \end{array}
\right.
$$
This system of equations is equivalent
to
$$
\left\{
  \begin{array}{l}
    x^2-\delta y^2=c  \\
     x=\frac{t-c-1}{2}.
  \end{array}
\right.
$$
Hence, $t\in I_w$ if and only if
the equation
\begin{equation}\label{-3 in NSQ}
\delta y^2=
\frac{(t-c-1)^2}{4}-c,
\end{equation}
with the variable
$y$ has solutions, that is, $
\frac{t-c-1}{2}\in D_c$.
From Lemma \ref{D_c}, we have
$$
\#I_w=
\left\{
  \begin{array}{ll}
    \frac{q+3}{2}, & c\in \textrm{SQ}_q \\
    \frac{q+1}{2}, & c\in \textrm{NSQ}_q.
  \end{array}
\right.
$$
\rm(2) From the Equation
(\ref{-3 in NSQ}), we have
$1\in I_1$ if and only if, $-\frac{3}{4}
\in \textrm{NSQ}_q$ or $-3=0$. Hence, $1\in I_1$ if and only if, $-3\in
\textrm{NSQ}_q$ or $p=3$.
\end{proof}

\begin{lemma}\label{HHw}
Let $w\in \mathbb{F}_{q^2}^{\times}$, where $w\not\in H_{+}$.
Then
$$
\#(H_{+}+H_{+}w)
=
\left\{
  \begin{array}{ll}
    \frac{(q+1)(q+3)}{2}, & \mathcal{N}(w)
\in \textrm{SQ}_q  \\
    \frac{(q+1)^2}{2}, & \mathcal{N}(w)
\in \textrm{NSQ}_q
  \end{array}
\right.
.
$$
\end{lemma}
\begin{proof}
Since $w\not\in H_{+}$,
$0\not\in H_{+}+H_{+}w$.
From Lemma \ref{Iw}, this lemma follows.
\end{proof}

\begin{theorem}\label{2H+}
Let $p$ be an odd prime. Then

{\rm (1)} $\#H_{+}^{(2)}=1+\frac{1}{2}
(q+1)^2$.

{\rm (2)} $\#H_{+}^{(2)}\backslash (H_{+}\cup \{0\})
=
\left\{
   \begin{array}{ll}
     \frac{(q+1)^2}{2}, & \hbox{$-3\in
\textrm{SQ}_q$} \\
    \frac{(q-1)(q+1)}{2}, & \hbox{
$-3\in \textrm{NSQ}_q \text{~or~} p=3$}
   \end{array}
 \right.
 .
$
\end{theorem}
\begin{proof}
{\rm (1)} For any $z_1 \in H_{+}$ and $z_2 \in H_{+}^{(2)}$, we have $z_1z_2\in H_{+}^{(2)}$.
Thus, one has the following decomposition
$$H_{+}^{(2)}=\cup_{t\in I_1 \setminus
\{0\}} H_{+}t \cup \{0\}.$$
By Lemma \ref{Iw},
$$
\#H_{+}^{(2)}=1+(q+1)
\frac{q+1}{2}=1+\frac{1}{2}
(q+1)^2.
$$

{\rm (2)} If $-3\in \textrm{NSQ}_q$ or p=3,
from Lemma \ref{Iw}, $1\in I_1$. Then
$H_{+}\cdot 1\subseteq H_{+}^{(2)}$. Hence,
$\#(H_{+}^{(2)}\setminus
(H_{+}\cup \{0\}))=1+\frac{1}{2}
(q+1)^2-(q+1)-1
$, that is, $\#(H_{+}^{(2)}\setminus
(H_{+}\cup \{0\}))=\frac{(q+1)(q-1)}{2}$.

If $-3\in \textrm{SQ}_q$, from from Lemma \ref{Iw}, $1\not\in I_1$. From the statement
(1) of this lemma, we have
$\#(H_{+}^{(2)}\setminus
(H_{+}\cup \{0\}))=\frac{1}{2}(q+1)^2$.
\end{proof}

\begin{theorem}\label{3H+}
Let $p$ be an odd prime. Then $H_{+}^{(3)}\cup \{0\}
=\mathbb{F}_{q^2}$.
\end{theorem}
\begin{proof}
Suppose that there exists
$z\in \mathbb{F}_{q^2}^{\times}$, with
$z\not\in H_{+}^{(3)}$. Then
$H_{+}z\cap H_{+}^{(3)}=\emptyset$. With the assumption, we first prove that $
(H_{+}+H_{+}z)\cap (H_{+}+H_{+})=\emptyset$. If
there exist $u_i\in H_{+}(i=1,2,3,4)$
such that
$$
u_1+u_2z=u_3+u_4,
$$
then $z=\frac{u_3}{u_2}
+\frac{u_4}{u_2}+(
-\frac{u_1}{u_2})\in H_{+}^{(3)}$.
This contradicts $z\not\in H_{+}^{(3)}$.
Hence, $(H_{+}+H_{+}z)\cap (H_{+}+H_{+})=\emptyset$.
Then $q^2=
\#\mathbb{F}_{q^2}\geq \#(H_{+}
+H_{+}z)+\#(H_{+}+H_{+})$.
From Lemma
\ref{HHw} and Theorem \ref{2H+} , we have
$$
q^2\geq \frac{(q+1)^2}{2}
+1+\frac{1}{2}(q+1)^2=
1+(q+1)^2.
$$
This makes a contradiction. Hence, this theorem follows.
\end{proof}

\begin{theorem}\label{subset H+}
Let $p$ be an odd prime with $p\equiv 1 \text{ or } -5 \mod 12$ and $k$ be any positive integer, or $p\ge 5$ and $k$ be an even positive integer. Let $H_{+}$ be defined as (\ref{H+}). Then, the expansion critical index and the expansion limit index of subset $H_{+}$ are equal 2 and 3, respectively.
\end{theorem}

\begin{proof}
Note that $H_{+}=q+1$. If  $p$ is an odd prime with $p\equiv 1 \text{ or } -5 \mod 12$ and $k$ any positive integer, or $p\ge 5$ and $k$ an even positive integer, then $-3\in \textrm{SQ}_q$. From Theorem \ref{2H+} and Theorem \ref{2 and 3},
the expansion critical index equals 2. By Theorem \ref{3H+} and Theorem \ref{2 and 3}, the expansion limit index equals 3. This completes the proof.
\end{proof}

\subsection{Some results on subset sums for $H_{-}$}
\begin{lemma}\label{H2SQ}
let $p$ be an odd prime,
$(a,b)\in \mathbb{F}_q^2$, and $t=ab$,   where
$(a,b)\neq (0,0)$.
Then $(a,b)\in H_{-}+H_{-}$
if and only if $t\neq 0$ and $(\frac{t}{2}-1)^2
-1\in \textrm{SQ}_q\cup \{0\}$.
\end{lemma}
\begin{proof}
Obviously, for $(a,b)\in
H_{-}^{(2)}\setminus \{(0,0)\}$, we have
$t=ab\neq 0$. From the definition of
$H_{-}^{(2)}$, we have
$(a,b)\in H_{-}^{(2)}\setminus \{(0,0)\}$
if and only if the following system of equations
\begin{equation}
\left\{
  \begin{array}{l}
    x+y=a   \\
    \frac{1}{x}+\frac{1}{y}=b
  \end{array}
\right.
\end{equation}
has solutions. This  is equivalent to the following equation
$$
\lambda^2- a \lambda+\frac{a}{b}=0.
$$
has solution.
Since the determinant of this quadratic equation
is $\Delta=a^2-\frac{4a}{b}
=\frac{ab}{b^2}(ab-4)
=\frac{4((\frac{t}{2}-1)^2-1)}{b^2}$, this quadratic
equation has solutions if and only if
$(\frac{t}{2}-1)^2-1 \in \textrm{SQ}_q
\cup\{0\}$. From the above discussion,
this lemma follows.
\end{proof}

\begin{lemma}\label{num-ab}
Let $p$ be an odd prime greater than 3.
Then
$\#\{ab: (a,b)\in H_{-}^{(2)}\setminus
\{(0,0)\}\}=\frac{q-1}{2}$.
\end{lemma}
\begin{proof}
From Lemma \ref{D_c}, we have
$$
\#\{x\in \mathbb{F}_q: x^2-1\in
\textrm{SQ}_q \cup
\{0\}\}=\frac{q+1}{2}.
$$
From Lemma \ref{H2SQ}, we have
$$
\#\{ab: (a,b)\in H_{-}^{(2)}\setminus
\{(0,0)\}\}=\frac{q-1}{2}.
$$
\end{proof}

\begin{theorem}\label{2H-}
Let $p$ be an odd prime.
Then

{\rm (1)} $\#H_{-}^{(2)}=1+\frac{(q-1)^2}{2}$;

{\rm (2)}
$\{0\}\in H_{-}^{(2)}$ and
$H_{-}\cap H_{-}^{(2)}
=\left\{
   \begin{array}{ll}
     \emptyset, & -3\in \textrm{NSQ}_q  \\
      H_{-}, &  -3\in \textrm{SQ}_q \text{~or~} p=3
   \end{array}
 \right.
$;

{\rm (3)} $\#(H_{-}^{(2)}\setminus
(H_{-}\cup \{0\}))
=\left\{
   \begin{array}{ll}
     \frac{(q-1)^2}{2}, & -3\in \textrm{NSQ}_q \\
     \frac{(q-1)^2}{2}-(q-1), &
-3\in \textrm{SQ}_q \text{~or~} p=3
   \end{array}
 \right.
$.
\end{theorem}
\begin{proof}
{\rm (1)} From Lemma \ref{H2SQ} and Lemma \ref{num-ab}, this statement follows.

{\rm (2)} From $(0,0)
=(1,1)+(-1,-1)$, we have
$(0,0)\in H_{-}^{(2)}$. Obviously,
$H_{-}\cap H_{-}^{(2)}
=\emptyset$ or $H_{-}$. Further,
$H_{-}\cap H_{-}^{(2)}=H_{-}$ if and
only if $(1,1)\in H_{-}^{(2)}$. From Lemma
\ref{H2SQ}, $(1,1)\in H_{-}^{(2)}$
if and only if $-3\in \textrm{SQ}_q \cup \{0\}$. Hence
this statement follows.

{\rm (3)} From the statements (1) and
(2), this statement follows.
\end{proof}

\begin{lemma}\label{irreducible}
If $t\neq -1$, the polynomial
$P_t(x,y)=(xy-x-y)(x+y+t)+xy
\in \mathbb{F}_q[x,y]$ is absolutely irreducible.
\end{lemma}
\begin{proof}
Suppose that the polynomial
$P_t(x,y)$ is not absolute irreducible. Then
there exist polynomials $A(x,y)$ and
$B(x,y)$ with coefficients in
the algebraic closure of $\mathbb{F}_q$
such that $P_t(x,y)=AB$, where
$deg(A)=1$ and $deg(B)=2$.
Further, the product of the leading terms of $A$ and $B$ is $xy(x+y)$. Then we consider the following cases:

{\rm 1)} Case $A(x,y)=x+b$. From the symmetry $P_t(x,y)=P_t(y,x)$, we have
$y+b$ is also a factor of $P_t(x,y)$. Hence,
$\frac{P_t(x,y)}{(x+b)(y+b)}$
is a factor of $P_t(x,y)$ of the
form $x+y+a$;

{\rm 2)} Case $A(x,y)=y+b$. From the similar discuss, $P_t(x,y)$ has a factor of the form $x+y+a$;

{\rm 3)} Case $A(x,y)=x+y+a$.

From the above discussion,
$x+y+a$ must be a factor of $P_t(x,y)$. Then
$$
P_t(x,y)=(xy-x-y)(x+y+t)+xy\equiv 0 \mod x+y+a.
$$
We have
$$
(t-a+1)xy+a(t-a)\equiv 0\mod x+y+a.
$$
Hence,  $t-a+1=0$ and
$a(t-a)=0$.   Hence, $a=0$ and $t=-1$.
It makes a contradiction. Hence, this lemma follows.
\end{proof}

\begin{theorem}\label{3H-}
Let $q$ be a power of odd prime with $q\ge 13$. Then
$H_{-}^{(3)}  \cup \{(0,0)\}
=\mathbb{F}_q^2 .
$
\end{theorem}
\begin{proof}
Obviously, $H_{-}\subseteq
H_{-}^{(3)}$. For any
$(a,b)\in \mathbb{F}_q^2\setminus \{(0,0)\}$, one has $a\neq 0$ or $b\neq 0$. Suppose that $b\neq 0$.
Then $(a,b)\in H_{-}^{(3)}$ if and only if the following system of equations
has solutions:
\begin{equation}
\left\{
  \begin{array}{l}
    x+y+z=a  \\
    \frac{1}{x}+\frac{1}{y}
+\frac{1}{z}=b.
  \end{array}
\right.
\end{equation}
We set $\frac{1}{0}=0$. This system
of equations is equivalent to the following system of equations:
$$
\left\{
  \begin{array}{l}
    bx+by+bz=ab  \\
    \frac{1}{bx}+\frac{1}{by}
+\frac{1}{bz}=1.
  \end{array}
\right.
$$
Hence, $(a,b)\in H_{-}^{(3)}$
if and only if $(ab,1)\in H_{-}^{(3)}$.
Hence, we just need to prove that
$(-t,1)\in H_{-}^{(3)}$ for any $t\in \mathbb{F}_q$.
If $t=-1$, $(1,1)=(1,1)+(1,1)+(-1,-1)
\in H_{-}^{(3)}$.
If $t\neq -1$, then
$(-t,1)\in H_{-}^{(3)}$ if and only if the following system of equations
has solutions:
$$
\left\{
  \begin{array}{l}
    x+y+z=-t \\
    \frac{1}{x}+\frac{1}{y}
+\frac{1}{z}=1.
  \end{array}
\right.
$$
From this system of equations, we have
$$
\frac{1}{x}+\frac{1}{y}
-\frac{1}{x+y+t}=1.
$$
Further, we obtain
$$
(x+y+t)(xy-x-y)+xy=0.
$$
This equation has solutions if and only if
the following affine curve
$$E_t: (x+y+t)(xy-x-y)+xy=0,$$ defined over
$\mathbb{F}_q$ has rational points.
The correspondence projective curve are
$$\overline{E}_t: (X+Y+tZ)(XY-ZX-YZ)+XYZ=0.$$
Then, we denote the $\mathbb{F}_q$-rational points sets associated with $E_t$ and $\overline{E}_t$ respectively by
$$E_t(\mathbb{F}_{q})=\{(x,y)\in \mathbb{F}_q^2: (x+y+t)(xy-x-y)+xy=0 \}$$
and
$$\overline{E}_t(\mathbb{F}_{q})=\{(X:Y:Z)\in \mathbb{P}_{\mathbb{F}_q}^2: (X+Y+tZ)(XY-ZX-YZ)+XYZ=0 \},$$
where $\mathbb{P}_{\mathbb{F}_q}^2$ denotes the projective space of dimension $2$ over $\mathbb{F}_q$
. The notation $(X:Y:Z)$ indicates a projective point, which is the same point as $(\lambda X:\lambda Y:\lambda Z)$ for any
$\lambda \neq 0$. Thus, affine solutions can be recovered by taking $\lambda=Z^{-1}$; except for solutions $(X:Y:0)$, which are the points at the infinite.\\
From Lemma \ref{irreducible}, $\overline{E}_t$ is an absolutely irreducible polynomial curves of degree $3$. Hasse-Weil bound states that
\begin{equation}\label{Et}
|\#\overline{E}_t(\mathbb{F}_q)-(q+1)|\leq 2\sqrt{q}.
\end{equation}
The only three projective
points $(X:Y:Z)$ with $Z=0$ are
$\{(1:0:0), (0:1:0), (1:-1:0)\}$.
The only three affine points with
$xy(x+y+t)=0$ are points in the set
$G_t=\{(0,0), (0,-t), (-t,0)\}$.
Hence,
$$
\#(E_t(\mathbb{F}_q)\setminus G_t)=\#(\overline{E}_t(\mathbb{F}_q))-6.
$$
From
Equation (\ref{Et}), if $q\ge 13$, we have $\#(\overline{E}_t(\mathbb{F}_q))-6 >0$. Hence, $(-t,1)\in H_{-}^{(3)}$, which concludes the proof.
\end{proof}

\begin{remark}
For $q=3,5,7,9, \text{ or } 11$, we can verify that $H_{-}^{(3)} \cup  \{(0,0)\}  \subsetneqq \mathbb{F}_q^2$.
\end{remark}

\begin{theorem}\label{subset H-}
Let $p$ be an odd prime with $p\equiv -1 \text{ or } 5 \mod 12$, $k$ be odd integer, and $q=p^k>12$. Let $H_{-}$ be defined as (\ref{H-}). Then, the expansion critical index and the expansion limit index of subset $H_{-}$ are equal to 2 and 3, respectively.
\end{theorem}

\begin{proof}
Note that $H_{-}=q-1$. If  $p$ is an odd prime with $p\equiv -1 \text{ or } 5 \mod 12$ and $k$ odd integer, then $-3\in \textrm{NSQ}_q$. From Theorem \ref{2H-} and Theorem \ref{2 and 3},
the expansion critical index is 2. By Theorem \ref{3H-} and Theorem \ref{2 and 3}, the expansion limit index is 3. This completes the proof.
\end{proof}

\section{Cayley graphs from $H_{+}$ and $H_{-}$}\label{sec:Cayley}
In this section, we will determined the eigenvalues of the Cayley graphs $Cay(\mathbb{F}_{q^2};H_{+})$ and $Cay(\mathbb{F}_{q}^2;H_{-})$. Then,
we will show the related graphs are Ramanujan or almost Ramanujan. It solves a conjecture proposed by Camarero and Mart\'{i}nez \cite{CM16}.

\subsection{Cayley graph  $Cay(\mathbb{F}_{q^2};H_{+})$}
\begin{theorem}\label{graph H+}
Let $p$ be an odd prime with $p\equiv 1 \text{ or } -5 \mod 12$ and $k$ be any positive integer, or $p\ge 5$ and $k$ be an even positive integer. Let $H_{+}$ be defined as (\ref{H+}). Then, the error correction capacity and the diameter of Cayley graph $Cay(\mathbb{F}_{q^2};H_{+})$ are 2 and 3, respectively.
\end{theorem}
\begin{proof}
This theorem follows from Theorem \ref{subset H+}, Theorem \ref{three quantities 0} and Theorem \ref{three quantities 1}.
\end{proof}

\begin{lemma}\label{UF}
Let $q=p^k$, with $p$ an odd prime. Then, for any $\alpha\in \mathbb{F}_{q^2}^{\times}$, we have
$$
\sum_{u\in H_{+}}
\zeta_p^{Tr_1^{2k}(\alpha u)}
=-\sum_{c\in \mathbb{F}_q^{\times}}
\zeta_p^{Tr_1^k(c+\frac{\mathcal{N}(\alpha)}{
c})}.
$$
\end{lemma}
\begin{proof}
Let $\alpha=a+\sqrt{\delta}b$,
$S=\sum_{u\in H_{+}}\zeta_p^{Tr_1^{2k}(\alpha u)}$,
and $C_1
=\{(x,y)\in \mathbb{F}_q^2:
x^2-\delta y^2=1\}$. Then
$$
S=\sum_{(x,y)\in C_1}\zeta_p^{Tr_1^{2k}(
(a+\sqrt{\delta}b)(x+\sqrt{\delta}y))}.
$$
From
$(a+\sqrt{\delta}b)(x+\sqrt{\delta}y)
=ax+\delta by+(ay+bx)\sqrt{\delta}$,
we have
$Tr_k^{2k}((a+\sqrt{\delta}b)
(x+\sqrt{\delta}y))
=2ax+2\delta by$ and
\begin{align*}
S=&\sum_{(x,y)\in C_1}\zeta_p^{Tr_1^{k}(
2ax+2\delta by)}\\
=& p^{-k}
\sum_{(x,y)\in \mathbb{F}_{q}^2}\zeta_p^{Tr_1^{k}(
2ax+2\delta by)}
\sum_{c\in \mathbb{F}_{q}}\zeta_p^{Tr_1^{k}(
c(x^2-\delta y^2-1))}\\
=&
p^{-k}\sum_{c\in \mathbb{F}_{q}}\zeta_p^{Tr_1^{k}(
-c)}
\sum_{x\in \mathbb{F}_{q}}\zeta_p^{Tr_1^{k}(
cx^2+2ax)}
\sum_{y\in \mathbb{F}_{q}}\zeta_p^{Tr_1^{k}(
-c\delta y^2+2\delta by)}\\
=& p^{-k}\sum_{c\in \mathbb{F}_{q}^\times}\zeta_p^{Tr_1^{k}(
-c)}
\sum_{x\in \mathbb{F}_{q}}\zeta_p^{Tr_1^{k}(
cx^2+2ax)}
\sum_{y\in \mathbb{F}_{q}}\zeta_p^{Tr_1^{k}(
-c\delta y^2+2\delta by)}.
\end{align*}
From Lemma \ref{cx2a}, we have
\begin{align*}
S=& p^{-k}p^{*k}
\eta(-\delta)\sum_{c\in \mathbb{F}_{q}^\times}\zeta_p^{Tr_1^{k}(
-c+\frac{\delta b^2-a^2}{c})}\\
=& -\eta_0^k(-1)\eta(-1) \sum_{c\in \mathbb{F}_{q}^\times}\zeta_p^{Tr_1^{k}(
c+\frac{a^2-\delta b^2}{c})}
\\
=& -\sum_{c\in \mathbb{F}_{q}^\times}\zeta_p^{Tr_1^{k}(
c+\frac{a^2-\delta b^2}{c})}
\;\;\;(\text{from} \;\eta(-1)=\eta_0^k(-1))\\
=& -\sum_{c\in \mathbb{F}_{q}^\times}\zeta_p^{Tr_1^{k}(
c+\frac{\mathcal{N}(\alpha)}{c})},
\end{align*}
which ends the proof.
\end{proof}

\begin{theorem}
Let $q=p^k$, with $p$ an odd prime. Then, $Cay(\mathbb{F}_{q^2};H_{+})$ is a Ramanujan graph.
\end{theorem}
\begin{proof}
Note that the characters of $\mathbb{F}_{q^2}$ are given by $\chi_{\alpha}(x)=\zeta_p^{Tr^{2k}_1(\alpha x)}$ as $\alpha$ ranges over $\mathbb{F}_{q^2}$. From Theorem \ref{eigenvelue}, all non-trivial eigenvalues of $Cay(\mathbb{F}_{q^2};H_{+})$ are $\lambda_{\chi_{\alpha}}=\sum_{u\in H_{+}} \chi_{\alpha}(u)$, with $\alpha \neq 0$.  When $\alpha \neq 0$, form Lemma \ref{UF} and Theorem \ref{Kloost}, one has
$$
|\lambda_{\chi_{\alpha}}|\le 2 \sqrt{q}.
$$
From Lemma \ref{c=q+1}, we have
$$
|\lambda_{\chi_{\alpha}}|\le 2 \sqrt{\# H_{+}-1}.
$$
Hence, $Cay(\mathbb{F}_{q^2};H_{+})$ is a Ramanujan graph. This completes the proof.
\end{proof}

\begin{remark}
When $q=p$, $p\equiv 3 \mod 4$, we can identity $\mathbb{F}_{p^2}$ as $\mathbb{F}_p[\sqrt{-1}]$. Thus, $Cay(\mathbb{F}_{p^2};H_{+})$ are
same as the graphs $\mathcal{G}_p$ defined by Camarero and Mart\'{\i}nez in \cite{CM16}. It solves the Conjecture $31$ in \cite{CM16}.
\end{remark}

\subsection{Cayley graph  $Cay(\mathbb{F}_q^2;H_{-})$}

\begin{theorem}\label{graph H-}
Let $p$ be an odd prime with $p\equiv -1 \text{ or } 5 \mod 12$, $k$ be an odd integer, and  $q=p^k>12$. Let $H_{-}$ be defined as (\ref{H-}). Then, the error correction capacity and the diameter of subset $Cay(\mathbb{F}_q^2;H_{-})$ are 2 and 3, respectively.
\end{theorem}
\begin{proof}
This theorem follows from Theorem \ref{subset H-}, Theorem \ref{three quantities 0} and Theorem \ref{three quantities 1}.
\end{proof}

\begin{theorem}\label{almost Ramanujan}
Let $q=p^k$, with $p$ an odd prime. Then, all non-trivial eigenvalues $\lambda$ of $Cay(\mathbb{F}_{q}^2;H_{-})$ satisfy
$$
|\lambda| \le 2 \sqrt{\#  H_{-}  +1}.
$$
\end{theorem}
\begin{proof}
Note that the characters of $\mathbb{F}_{q}^2$ are given by $\chi_{a,b}(x)=\zeta_p^{Tr^{k}_1(a x+b y)}$ as $(a,b)$ ranges over $\mathbb{F}_{q}^2$.
From Theorem \ref{eigenvelue}, all non-trivial eigenvalues of $Cay(\mathbb{F}_{q}^2;H_{-})$ are $\lambda_{\chi_{a,b}}=\sum_{(x,y)\in H_{-}} \chi_{a,b}(x,y)$, with $(a,b) \neq (0,0)$. When $(a,b) \neq (0,0)$, form Theorem \ref{Kloost}, one has
$$
|\lambda_{\chi_{a,b}}|\le 2 \sqrt{q}.
$$
From $\# H_{-}=q-1$, we have
$$
|\lambda_{\chi_{\alpha}}|\le 2 \sqrt{\# H_{-}+1}.
$$
This completes the proof.
\end{proof}
\begin{remark}
From Theorem \ref{almost Ramanujan}, Cayley graph $Cay(\mathbb{F}_{q}^2;H_{-})$ are almost Ramanujan.
\end{remark}

\section{$2$-Quasi-Perfect Lee Codes from $H_{+}$ and $H_{-}$}\label{sec:quasi-perfect}
In this section, we will present the two classes of two-quasi-perfect Lee codes. If $p\equiv -5 \mod 12$ and $k=1$, then $-1$ is a nonsquare in $\mathbb{F}_p^{\times}$ and $\mathbb{F}_{p^2}=\mathbb{F}_p[\sqrt{-1}]$. Thus, $\mathcal{N}(x+y\sqrt{-1})=x^2+y^2$. In this case, the codes defined in \cite{CM16} is same as $\mathcal{C}(\mathbb{F}_{p^2};H_{+})$. If $p\equiv 5 \mod 12$ and $k=1$, then $-1$ is a square in $\mathbb{F}_p^{\times}$ and $x^2+y^2=(x+y\sqrt{-1})(x-y\sqrt{-1})$. In this case, the codes defined in \cite{CM16} is same as $\mathcal{C}(\mathbb{F}_{p}^2;H_{-})$. Hence, our constructions generalize the constructions presented in \cite{CM16}.

\subsection{Quasi-Perfect Lee Codes $\mathcal{C}(\mathbb{F}_{q^2};H_{+})$}
Let $\{\mathbf{e}_1,\cdots, \mathbf{e}_{2k}\}$ be a basis of $\mathbb{F}_{q^2}$ over $\mathbb{F}_p$. Let $H_{+}=\{\pm \beta_1,\cdots,\pm \beta_{n}\}$, where $n=\frac{q+1}{2}$ and $\beta_j=\sum_{i=1}^{2k} h_{i,j} e_i$. Then, from the definition of $\mathcal{C}(\mathbb{F}_{q^2};H_{+})$, the parity-check matrix of $\mathcal{C}(\mathbb{F}_{q^2};H_{+})$ is
\begin{align}
  \label{matrix:H+}
  \mathbf{M}_{+}=\left(\begin{array}{cccc}
      h_{1,1} & h_{1,2} & \dots & h_{1,n}\\
      h_{2,1} & h_{2,2} & & h_{2,n}\\
      \vdots & \vdots &  & \vdots\\
      h_{2k,1} &h_{2k,2} & & h_{2k,n}
   \end{array}\right).
\end{align}
By Theorem \ref{3H+}, the rank of $\mathbf{M}_{+}$ is $2k$. Thus, the dimension of $\mathcal{C}(\mathbb{F}_{q^2};H_{+})$ is $n-2k$. From Theorem \ref{subset H+}, Theorem \ref{three quantities 0} and Theorem \ref{three quantities 1},  we have the following theorem.
\begin{theorem}\label{code H+}
Let $p$ be an odd prime with $p\equiv 1 \text{ or } -5 \mod 12$ and $k$ be any positive integer, or $p\ge 5$ and $k$ be an even positive integer. Let $H_{+}$ be defined as (\ref{H+}). Then, the linear code $\mathcal{C}(\mathbb{F}_{q^2};H_{+})$ is a $2$-quasi-perfect Lee code over $\mathbb{F}_p^n$ with dimension $\frac{q+1}{2}-2k$, where $n=\frac{q+1}{2}$. Moreover, $\mathbf{M}_{+}$ is the parity-check matrix of $\mathcal{C}(\mathbb{F}_{q^2};H_{+})$.
\end{theorem}
\begin{example}
Let $p=13$, $k=1$ $n=7$ and $\mathbb{F}_{13^2}=\mathbb{F}_{13}[\sqrt{2}]$. Then, the codes over $\mathbb{Z}_{13}^{7}$ defined by the parity-check matrix
\begin{align}
  \left(\begin{array}{ccccccc}
      1 & 9 & 5 & 3 & 10 & 8 & 4\\
      0 & 1 & 5 & 11 & 11 & 5 & 1\\
   \end{array}\right).
\end{align}
results in a $2$-quasi-perfect $13$-ary Lee codes. This codes has $p^{n-2}=371~293$ codewords.
\end{example}

\begin{remark}
If $p$ is an odd prime with $p\equiv -1 \text{ or } +5 \mod 12$ and $k$ is an odd integer, $Cay(\mathbb{F}_{q^2};H_{+})$ only contains $2n^2+1$ vertices at distance $2$
or less from vertex $0$, where $n=\frac{q+1}{2}$. In this case, although $\mathcal{C}(\mathbb{F}_{q^2};H_{+})$ is not a $2$-error correcting code, it is very close to it, since only $2n$
syndromes cannot be corrected.
\end{remark}

\subsection{Quasi-Perfect Lee Codes $\mathcal{C}(\mathbb{F}_{q}^2;H_{-})$}
Let $\{\mathbf{e}_1,\cdots, \mathbf{e}_{2k}\}$ be a basis of $\mathbb{F}_{q}^2$ over $\mathbb{F}_p$. Let $H_{-}=\{\pm \beta_1,\cdots,\pm \beta_{n}\}$, where $n=\frac{q-1}{2}$ and $\beta_j=\sum_{i=1}^{2k} h_{i,j} e_i$. Then, from the definition of $\mathcal{C}(\mathbb{F}_{q}^2;H_{-})$, the parity-check matrix of $\mathcal{C}(\mathbb{F}_{q}^2;H_{-})$ is
\begin{align}
  \label{matrix:H-}
  \mathbf{M}_{-}=\left(\begin{array}{cccc}
      h_{1,1} & h_{1,2} & \dots & h_{1,n}\\
      h_{2,1} & h_{2,2} & & h_{2,n}\\
      \vdots & \vdots &  & \vdots\\
      h_{2k,1} &h_{2k,2} & & h_{2k,n}
   \end{array}\right).
\end{align}
By Theorem \ref{3H-}, the rank of $\mathbf{M}_{-}$ is $2k$. Thus, the dimension of $\mathcal{C}(\mathbb{F}_{q}^2;H_{-})$ is $n-2k$. From Theorem \ref{subset H-}, Theorem \ref{three quantities 0} and Theorem \ref{three quantities 1},  we have the following theorem.
\begin{theorem}\label{graph H-}
Let $p$ be an odd prime with $p\equiv -1 \text{ or } 5 \mod 12$,  $k$ be an odd integer, and $q=p^k>12$. Let $H_{-}$ be defined as (\ref{H-}). Then,the linear code $\mathcal{C}(\mathbb{F}_{q}^2;H_{-})$  is a $2$-quasi-perfect Lee code over $\mathbb{F}_p^n$ with dimension $\frac{q-1}{2}-2k$, where $n=\frac{q-1}{2}$. Moreover, $\mathbf{M}_{-}$ is the parity-check matrix of $\mathcal{C}(\mathbb{F}_{q}^2;H_{-})$.
\end{theorem}

\begin{example}
Let $p=23$, $k=1$ and $n=11$. Then, the codes over $\mathbb{Z}_{23}^{11}$ defined by the parity-check matrix
\begin{align}
  \left(\begin{array}{ccccccccccc}
      1 & 2 & 3 & 4 & 5 & 6 & 7 & 8 & 9 & 10 & 11\\
      1 & 12 & 8 & 6 & 14 & 4 & 10 & 3 & 18 & 7 & 21\\
   \end{array}\right).
\end{align}
results in a $2$-quasi-perfect $23$-ary Lee codes. This codes has $p^{n-2}=23^{9}$ codewords.
\end{example}

\begin{remark}
If $p$ is an odd prime with $p\equiv 1 \text{ or } -5 \mod 12$  and $q=p^k>12$, or, if $p$ is an odd prime with $p\equiv -1 \text{ or } 5 \mod 12$
, $k$ is an even integer and $q=p^k>12$
,$Cay(\mathbb{F}_{q}^2;H_{-})$ only contains $2n^2+1$ vertices at distance $2$
or less from vertex $0$, where $n=\frac{q-1}{2}$. In this case, although $\mathcal{C}(\mathbb{F}_{q}^2;H_{-})$ is not a $2$-error correcting code, it is very close to it, since only $2n$
syndromes cannot be corrected.
\end{remark}

Now, Let us give some considerations on the quality of the constructed $2$-quasi-perfect Lee codes. Note, that, since the Lee sphere of radius $2$ contains $\#B^n_2=2n^2+2n+1$ words, the graph induced by any $2$-quasi-perfect linear code has at least $2n^2+2n+1$ vertices. The graphs $Cay(\mathbb{F}_{q^2};H_{+})$ and $Cay(\mathbb{F}_{q}^2;H_{-})$ constructed in this paper have $q^2$ vertices. Therefore, for the case $Cay(\mathbb{F}_{q^2};H_{+})$,  the number of vertices is $q^2=4n^2-4n+1=2 \# B^n_2 -8n+1$, where $n=\frac{q+1}{2}$. Also, for the case
$Cay(\mathbb{F}_{q}^2;H_{-})$, the number of vertices is $q^2=4n^2+4n+1=2 \# B^n_2 -1$, where $n=\frac{q-1}{2}$. Thus, the reached vertices are asymptotically the double of those that would be reached in the graph associated to a perfect code. In other words, when $q=p$, the density of the codes presented is $\frac{1}{p^2}$.

\section*{Concluding Remarks}
In this paper, we firstly  survey the relationships between subset sums, Cayley graphs, and Lee linear codes and present some results. Next, we suggest
an original approach to construct two classes of $2$-quasi-perfect Lee codes defined over the space $\mathbb{Z}_p^n$ for $n=\frac{p^k+1}{2}$ $(\text{with} ~p\equiv 1, -5 \mod 12 \text{ and } k \text{ is any integer}, \text{ or } p\equiv -1, 5 \mod 12 \text{ and } k \text{ is an even integer})$ and $n=\frac{p^k-1}{2}$ $(\text{with }p\equiv -1, 5 \mod 12,  k \text{ is an odd integer}\text{ and } p^k>12)$, where $p$ is an  odd prime. To this end,  subsets from some quadratic curves over finite fields have been considered. The obtained  classes of codes contain the quasi-perfect Lee codes constructed by Camarero and Mart\'{\i}nez in \cite{CM16}. We therefore generalize some results presented in \cite{CM16} and solve a conjecture proposed in \cite{CM16} by proving that the related Cayley graphs are Ramanujan or almost Ramanujan. Our approach can be used for constructing other classes of quasi-perfect Lee codes which have many practical applications.

\ifCLASSOPTIONcaptionsoff
  \newpage
\fi

\end{document}